\documentclass[submission,copyright,creativecommons]{eptcs}
\providecommand{\event}{FICS 2015} % Name of the event you are submitting to
\usepackage{breakurl}             % Not needed if you use pdflatex only.

\usepackage{amsthm}
\usepackage{amsmath}
\usepackage{amssymb}
\usepackage{bussproofs}
\usepackage{bm}
\usepackage{lscape}
\usepackage{rotating}
\usepackage{multicol}
\usepackage{float}

\makeatletter
\numberwithin{equation}{section}
\numberwithin{figure}{section}
\theoremstyle{plain}
\newtheorem{thm}{Theorem}
  \theoremstyle{definition}
  \newtheorem{definition}[thm]{Definition}
  \theoremstyle{plain}
  \newtheorem{fact}[thm]{Fact}
  \theoremstyle{definition}
  \newtheorem{remark}[thm]{Remark}
  \theoremstyle{plain}
  \newtheorem{lem}[thm]{Lemma}
 \theoremstyle{definition}
  \newtheorem{example}[thm]{Example}

\makeatother

\title{Disjunctive form and the modal $\mu$ alternation hierarchy}
\author{Karoliina Lehtinen
\institute{Laboratory for Foundations of Computer Science}
\institute{University of Edinburgh}
\email{M.K.Lehtinen@sms.ed.ac.uk}}
\def\titlerunning{Disjunctive form and the modal $\mu$ alternation hierarchy}
\def\authorrunning{M.K. Lehtinen}

\newenvironment{scprooftree}[1]%
  {\gdef\scalefactor{#1}\begin{center}\proofSkipAmount \leavevmode}%
  {\scalebox{\scalefactor}{\DisplayProof}\proofSkipAmount \end{center} }

\begin{document}
\maketitle

\global\long\def\dmod{{\rightarrow}}
\begin{abstract}
This paper studies the relationship between disjunctive form,
a syntactic normal form for the modal $\mu$ calculus, and the alternation hierarchy.
First it shows that all disjunctive formulas which
have equivalent tableau have the same syntactic alternation depth.
However, tableau equivalence only preserves alternation depth for
the disjunctive fragment: there are disjunctive formulas with arbitrarily
high alternation depth that are tableau equivalent to alternation-free
non-disjunctive formulas. Conversely, there are non-disjunctive formulas
of arbitrarily high alternation depth that are tableau equivalent
to disjunctive formulas without alternations. This answers negatively
the so far open question of whether disjunctive form preserves alternation
depth. The classes of formulas studied here illustrate a previously
undocumented type of avoidable syntactic complexity which may contribute
to our understanding of why deciding the alternation hierarchy is
still an open problem.
\end{abstract}

\section{Introduction}

The modal $\mu$ calculus \cite{bradfield2007modal}, $L_{\mu}$, is a
modal logic augmented with its namesake least fixpoint operator $\mu$ and the dual greatest
fixpoint operator, $\nu$. Alternating between these two operators
gives the logic its great expressivity \cite{bradfield-strict} while both model checking and
satisfiability remain pleasingly decidable. The complexity of model
checking is, at least currently, tied to the number of such alternations,
called the alternation depth of the formula being checked \cite{jurdzinski2000small}.
The problem of deciding the least number of alternations required to express a property, also known as
the Rabin-Mostowski index problem, is a long standing open problem.

Disjunctive normal form is a syntactic restriction on $L_\mu$
formulas which first appeared in \cite{janin1995automata} and was then 
used as a tool for proving completeness of Kozen's axiomatization \cite{walukiewicz2000completeness}.
 It is based on the tableau decomposition of a formula which forces it to be in many ways well-behaved, 
making it a useful tool for various manipulations. For instance, satisfiability and synthesis are straight-forward for disjunctive
formulas. In \cite{DAgostinoHollenberg} it is used to analyse modal $L_\mu$ from a logician's perspective.
 More recently, disjunctive form was found to allow for simple formula optimisation:
if a formula is equivalent to a formula without greatest fixpoints,
then such a formula is easily produced by simple syntactic manipulation
on the disjunctive form of the formula \cite{me-CSL-2015}.

Each of these results uses the fact that any formula can be effectively
transformed into an equivalent disjunctive formula with the same tableau
-- indeed, disjunctive form is perhaps the closest one gets to a canonical normal
form for $L_\mu$. The transformation itself, described in \cite{janin1995automata},
is involved and it has so far been an open question whether it preserves
the alternation depth of formulas. If this was the case, it would
be sufficient to study the long-standing open problem of the decidability
of the alternation hierarchy on this well-behaved fragment. \\

 In this paper, we show that although the disjunctive fragment of $L_\mu$
is itself well-behaved with respect to the alternation
hierarchy, the transformation into it does not preserve alternation
depth.

The transformation into disjunctive form takes the tableau decomposition of a formula, and  produces a disjunctive formula that generates the same tableau.
 The first contribution of this paper is to show that all disjunctive formulas generating the same tableau have the same alternation depth.
This result brings some clarity to the transformation into disjunctive form since one of the more
difficult steps of the construction is representing the parity of infinite paths of the tableau with a finite priority assignment.
The result presented here means that all valid choices are just as good, as all yield a disjunctive formula
of the same alternation depth. As a result, the alternation hierarchy is decidable
for the disjunctive fragment of $L_\mu$ with respect to tableau
equivalence, a stricter notion of equivalence than semantic equivalence, as defined in \cite{walukiewicz2000completeness}.

The second contribution of this paper is to show that this does not extend to non-disjunctive formulas. Not only
does tableau equivalence not preserve alternation depth in general, but 
the alternation depth of a formula does not guarantee \textit{any} upper bound on the alternation depth
of equivalent disjunctive formulas. Indeed, for arbitrarily large $n$, there
are formulas with a single alternation which are tableau equivalent
only to disjunctive formulas with at least $n$ alternations.

Conversely, there are formulas of $L_\mu$ with arbitrarily large
alternation depth which are tableau equivalent to a disjunctive formula
without alternations. This shows that the alternation depths of tableau
equivalent formulas are only directly related within the disjunctive
fragment. \\

The signficance of these results in twofold. First, they outline the limits of what can be achieved using disjunctive form:
disjunctive form does not preserve alternation depth so despite being a useful tool for satisfiability-related problems, it
is unlikely to be of much help in contexts where the alternation depth of a formula matters, such as model-checking or 
formula optimisation beyond the first levels of the alternation hierarchy.

Secondly, and perhaps most significantly, these results impact our understanding of the alternation hierarchy.
This paper's results imply that deciding the alternation hierarchy for the disjunctive fragment of $L_\mu$, an open but
easier problem, is not sufficient for deciding the alternation hierarchy in the general case.
The counterexamples used to show this illustrate a previously undocumented type of accidental
complexity which appears to be difficult to identify. These may shed light on why deciding the alternation hierarchy is still
an open problem and examplify a category of formulas with unnecessary alternations which need to be tackled with novel methods.

\paragraph{Related work}
Deciding the modal $\mu$ alternation hierarchy is exactly equivalent to deciding the Rabin-Mostowski index of alternating parity automata.
The corresponding problem has also been studied for automata operating on words \cite{Carton1999} and automata which are deterministic \cite{det-rabin,flowers},
or non-deterministic \cite{non-det-rabin,game-aut-rabin} rather than alternating. As will be highlighted throughout this paper,
many of the methods used here are similar to methods applied to different types of automata.

\section{Preliminaries}

\subsection{The modal $\mu$ calculus}

For clarity and conciseness, the semantics of $L_{\mu}$ are given directly
in terms of parity games. As is well documented in the literature,
this approach is equivalent to the standard semantics \cite{bradfield2007modal}.
The following definitions are fairly standard, although we draw the reader's attention to
the use of the less typical modality $\dmod\mathcal{B}$ in the syntax of $L_\mu$ and
the unusual but equivalent definition of alternation depth.

\begin{definition}
($L_\mu$) Given a set of atomic propositions $Prop=\{P,Q,...\}$
and a set of fixpoint variables $Var=\{X,Y,...\}$ , the syntax of
$L_{\mu}$ is given by:
\[ 
   \phi:= \top 
   \mbox{ $|$ }\bot
   \mbox{ $|$ } P
   \mbox{ $|$ }\neg P
   \mbox{ $|$ } X
   \mbox{ $|$ }\phi\wedge\phi
   \mbox{ $|$ }\phi\vee\phi
   \mbox{ $|$ }\dmod\mathcal{B}\mbox{ where \ensuremath{\mathcal{B}} is a set of formulas }
   \mbox{ $|$ }\mu X.\phi
   \mbox{ $|$ }\nu X.\phi
\]

 The modality $\dmod\mathcal{B}$ replaces the more usual
modalities $\Diamond\phi$ and $\Box\phi$. If $\mathcal{B}$ is a set of formulas, $\dmod\mathcal{B}$
stands for $(\bigwedge_{\phi\in\mathcal{B}}\Diamond\phi)\wedge\square\bigvee_{\phi\in\mathcal{B}}\phi$:
every formula in $\mathcal{B}$ must be realised in some successor state and each successor state must realise at
least one of the formulas in $\mathcal{B}$.
The modalities $\Diamond\phi$ and $\Box\phi$ are expressed in this syntax by
$\dmod\{\phi,\top\}$ and $\dmod\{\phi\}\vee\dmod\bot$ respectively, where $\bot$ denotes the empty set.
\end{definition}

Without loss of expressivity, this syntax only allows for formulas in positive form:
negation is only applied to propositions. Furthermore, without loss of expressivity, but perhaps conciseness, we require all formulas to be guarded:
all fixpoint variables are within the scope of a modality within their binding formula. For the sake of clarity,
we restrict our study to the uni-modal case but expect the multi-modal
case to behave broadly speaking similarly. To minimise the use of brackets, the scope of fixpoint bindings
should be understood to extend as far as possible.
\begin{definition}
\textit{(Structures)} A structure $\mathcal{M}=(S,s_{0},R,P)$ consists
of a set of states $S$, rooted at some initial state $s_{0}\in S$,
and a successor relation $R\subseteq S\times S$ between the states.
Every state $s$ is associated with a set of propositions $P(s)\subseteq Prop$
which it is said to satisfy.
\end{definition}
\begin{definition}

\textit{(Parity games)} A parity game is a potentially infinite two-player
game on a finite graph $\mathcal{G}=(V_{0},V_{1},E,v_{I},\Omega)$ of which
the vertices $V_{0}\cup V_{1}$ are partitioned between the two players
Even and Odd and annotated with positive integer priorities via $\Omega:V_{0}\cup V_{1}\rightarrow\mathbb{N}$.
The even player and her opponent, the odd player, move a token along
the edges $E\subseteq V_{0}\cup V_{1}\times V_{0}\cup V_{1}$ of the graph
starting from an initial position $v_{I}\in V_{0}\cup V_{1}$, each
choosing the next position when the token is on a vertex in their
partition. Some positions $p$ might have no successors in which case
they are winning for the player of the parity of $\Omega(p)$. A play
consists of the potentially infinite sequence of vertices visited
by the token. For finite plays, the last visited parity decides the
winner of the play. For infinite play, the parity of the highest priority
visited infinitely often decides the winner of the game: Even wins
if the highest priority visited infinitely often is even; otherwise
Odd wins. Note that since some readers may be used to an equivalent definition
using the lowest priority to define the winner, whenever possible,
``most significant'' will be used to indicate the highest priority.
\end{definition}
\begin{definition}
\textit{(Strategies)} A positional strategy $\sigma$ for one of the
players in $\mathcal{G}=(V_{0},V_{1},E,v_{I},\Omega)$ is a mapping
from the player's positions $s$, in $V_{0}$ for Even and in $V_{1}$ for Odd, in the game to a
successor position $s'$ such that $(s,s')\in E$. A play respects a player's strategy $\sigma$
if the successor of any position in the play belonging to the player
is the one dictated by $\sigma$. If $\sigma$ is Even's strategy
and $\tau$ is Odd's strategy then there is a unique play $\sigma\times\tau$
respecting both strategies from every position. The winner of the
parity game at a position is the player who has a strategy $\sigma$
, said to be a winning strategy, such that they win $\sigma\times\tau$
from that position for any counter-strategy $\tau$. The following
states that such strategies are sufficient: players do not need to
take into account the history of a play to play optimally.\end{definition}

\begin{fact}
Parity games are positionally determined: for every position either
Even or Odd has a winning strategy \cite{emerson1991tree}.
\end{fact}
 This means
that strategies gain nothing from looking at the whole play rather
than just the current position. As a consequence, we may take a strategy
to be memoryless: it is a mapping from a player's positions to a successor. \\

For any $L_\mu$ formula $\phi$ and a structure $\mathcal{M}$
we define a parity game $\mathcal{M}\times\phi$, constructed in polynomial
time, and say that $\mathcal{M}$ satisfies $\phi$, written $\mathcal{M}\models\phi$,
if and only if the Even player has a winning strategy in $\mathcal{M}\times\phi$. 
\begin{definition}
\textit{(Model-checking parity game)} For any formula $\phi$ of modal
$\mu$, and a model $\mathcal{M}$,
define a parity game $\mathcal{M}\times\phi$ with positions $(s,\psi)$
where $s$ is a state of $\mathcal{M}$ and $\psi$ is either a proper subformula
of $\phi$, or the formula $\bigvee\mathcal{B}$, or the formula $\Diamond\psi$ for
any $\dmod\mathcal{B}$ and $\psi\in\mathcal{B}$ in $\phi$.
The initial position is $(s_{0},\phi)$ where $s_{0}$ is the root of $\mathcal{M}$. Positions $(s,\psi)$
where $\psi$ is a disjunction or $\Diamond\psi'$ belong to Even
while conjunctions and positions $\dmod\mathcal{B}$ belong to Odd.
Other positions have at most one successor; let them be Even's although
the identity of their owner is irrelevant. There are edges from $(s,\psi\vee\psi')$
and $(s,\psi\wedge\psi')$ to both $(s,\psi)$ and $(s,\psi')$; from
$(s,\mu X.\phi)$ and $(s,\nu X.\phi)$ to $(s,\phi)$; from $(s,X)$
to $(s,\nu X.\psi)$ if $X$ is bound by $\nu$, or $(s,\mu X.\psi)$
if it is bound by $\mu$; finally, from $(s,\dmod\mathcal{B})$ to every 
$(s',\bigvee\mathcal{B})$ where $(s,s')$ is an edge in $\mathcal{M}$, and also to every
$(s,\Diamond\psi)$ where $\psi\in\mathcal{B}$ and from $(s,\Diamond\psi)$ to every $(s',\psi)$ where $(s,s')$
is an edge in the model $\mathcal{M}$. Positions $(s,P)$,$(s,\neg P),(s,\top)$
and $(s,\bot)$ have no successors. The parity function assigns an
even priority to $(s,\top)$ and also to $(s,P)$ if $s$ satisfies
$P$ in $\mathcal{M}$ and to $(s,\neg P)$ if $s$ does not satisfy
$P$ in $\mathcal{M}$; otherwise $(s,P)$ and $(s,\neg P$) receive
odd priorities, along with $(s,\bot)$. Fixpoint variables are given
positive integer priorities such that $\nu$-bound variables receive even
priorities while $\mu$-bound variables receive odd priorities. Furthermore,
whenever $X$ has priority $i$, $Y$ has priority $j$ and $i<j$,
$X$ must not appear free in the formula $\psi$ binding $Y$ in $\mu Y.\psi$
or $\nu Y.\psi$. In other words, inner fixpoints receive lower, less
significant priorities while outer fixpoint receive high priorities.
Other nodes receive the least priority used, 0 or 1.
\end{definition}
We now use parity games to define the semantics of $L_{\mu}$.
\begin{definition}
\textit{(Satisfaction relation)} A structure $\mathcal{M}$, rooted
at $s_{0}$ is said to satisfy a formula $\Psi$ of $L_{\mu}$, written
$\mathcal{M}\models\Psi$ if and only if the Even player has a winning
strategy from $(s_{0},\Psi)$ in $\mathcal{M}\times\Psi$ .
\end{definition}
Note that the definition of the model-checking parity game requires
a priority assignment to fixpoint variables in a formula that satisfies
the conditions that $\nu$-variables receive even priorities, $\mu$-variables
receive odd priorities and whenever $X$ has priority $i$, $Y$ has priority
$j$ and $i<j$, $X$ must not appear free in the formula $\psi$
binding $Y$ in $\mu Y.\psi$ or $\nu Y.\psi$. For any formula, there are several
valid assignments. For example, one could assign a distinct priority to every fixpoint,
with the highest priority going to the outermost bound fixpoint and the priorities decreasing
the further into the formula a fixpoint is bound. We further restrict a parity assignment to
be surjective into an initial fragment of $\mathbb{N}$: if a priority is unused, all greater priorities can be reduced by $2$.
We define the alternation depth of a formula to be the minimal valid assignment. Although variations
of this definition exists, our motivation is to match closely the alternations required in the model checking parity game.
\begin{definition}
Let a priority assignment be a function $\Omega:Var\rightarrow \{0...n\}$ for some integer $n$,
which is surjective on at least $\{1,...,n\}$,
such that if $\Omega(X)<\Omega(Y)$ then $X$ does not appear free
in the formula binding $Y$ and the parity of $\Omega(X)$ is even
for $\nu$-bound variables and odd for $\mu$-bound variables.
We don't require the priority $0$ to be used, but include it in the co-domain for simplicity.
 In this paper, we take the alternation depth of a formula to be the co-domain
of the least priority assignment of a formula. The correspondance with the priorities of the
model checking parity game should make it clear that this definition is equivalent to the more typical syntactic
ones in the literature, for example in \cite{bradfield2007modal}.
An alternation free formula is a formula which has both priority assignements with co-domain $\{0,1\}$ and $\{0,1,2\}$ where $0$ is not used.
\end{definition}
Deciding whether a formula is equivalent to a formula with smaller alternation
depth is a long standing open problem.

\subsection{Tableau decomposition}
\begin{definition}
\textit{(Tableau)} A tableau $\mathcal{T}=(T,L)$ of a formula $\Psi$
consists of a potentially infinite tree $T$ of which each node $n$
has a label $L(n)\subseteq\mathit{sf}(\Psi)$ where $\mathit{sf}(\Psi)$ is the
set of proper subformulas of $\Psi$. The labelling respects
the following tableau rules with the restriction that the modal rule is only applied where
no other rule is applicable. 
\begin{figure}[H]

\begin{multicols}{3}

\begin{prooftree}
\AxiomC{$\{ \Gamma , \phi , \psi \}$}
\RightLabel{$(\wedge)$}
\UnaryInfC{$\{ \Gamma, \psi \wedge \phi \}$}
\end{prooftree}

\begin{prooftree}
\AxiomC{$\{ \Gamma , \phi \}$}
\AxiomC{$\{ \Gamma , \psi \}$}
\RightLabel{$(\vee)$}
\BinaryInfC{$\{ \Gamma, \psi \vee \phi \}$}
\end{prooftree}

\begin{prooftree}
\AxiomC{$\{ \Gamma , \phi \}$}
\RightLabel{$(\sigma)$  with $\sigma \in \{\mu ,\nu \}$}
\UnaryInfC{$\{ \Gamma, \sigma X.\phi \}$}
\end{prooftree}
\end{multicols}
\begin{prooftree}
\AxiomC{$\{ \Gamma , \phi \}$}
\RightLabel{$(X)$ where $X$ is a fixpoint variable bound by $\sigma X.\phi$, with $\sigma\in\{\mu,\nu\}$}
\UnaryInfC{$\{ \Gamma,  X \}$}
\end{prooftree}
\begin{prooftree}
\AxiomC{$\{ \psi \}\cup \{ \bigvee \mathcal{B} | \dmod \mathcal{B} \in \Gamma, \mathcal{B}\neq\mathcal{B}' \} $ for every $\dmod \mathcal{B}' \in \Gamma, \psi \in \mathcal B' $}
\RightLabel{$(\dmod)$}
\UnaryInfC{$\{ \Gamma \}$}

\end{prooftree}

\end{figure}

Note that each branching
node is either a choice node, corresponding to a disjunction, or a
modal node. Although the rules only contain a binary disjunctive rule,
we may write, for the sake conciseness, a sequence of binary choice
nodes as a single step. Also note that when a modal rule is applied,
all formulas in a label are either modal formulas or literals, that is to say propositional variables and their negations.
The latter form the modal node's set of literal and are a
semantically important component of the tableau. An inconsistent set of literals is equivalent to $\bot$ and a node with such
a set of literals in its label has no successors.
\end{definition}

Sequences of subformulas along a path in the tableau are called traces and correspond to plays in
the model checking parity game. A $\mu$-trace is a trace winning for the Odd player.

\begin{definition}
\textit{($\mu$-trace)} Given an infinite branch in a tableau, that
is to say a sequence $n_{0}n_{1}...$ of nodes starting at the root,
where $n_{i+1}$ is a child of $n_{i}$, a trace on it is an infinite
sequence $f_{0}f_{1}...$ of formulas satisfying the following: each
formula is taken from the label of the corresponding node, $f_{i}\in L(n_{i})$
for all $i\geq0$; successive formulas $f_{i}$ and $f_{i+1}$ are
identical if $f_{i}$ is not the formula that the tableau rule
from $n_{i}$ to $n_{i+1}$ acts on; 
if the tableau rule from $n_{i}$
to $n_{i+1}$ is a disjunction, conjunction, or fixpoint binding elimination acting on $f_{i}$,
then $f_{i+1}$ is an immediate subformula of $f_{i}$;
if the tableau rule from $n_{i}$ to $n_{i+1}$ is a modality, then $f_i$ has to be a formula
 $\dmod\mathcal{B}$ and $f_{i+1}$ is either $\bigvee \mathcal{B}$ or a formula $\psi\in\mathcal{B}$;
if the tableau rule from $n_{i}$ to $n_{i+1}$ is a fixpoint regeneration acting
on the fixpoint variable $f_{i}$, then $f_{i+1}$ is the binding
formula for $f_{i}$. A trace is a $\mu$-trace if the most significant
fixpoint variable that regenerates infinitely often on it is a $\mu$-variable.
\end{definition}
Since labels are to be thought of as conjuncts, it is sufficient for
an infinite path in a tableau to allow one $\mu$-trace for the infinite
path to be winning for the Odd player.
\begin{definition}
\textit{(Parity of a path)} An infinite path in a tableau is said to be even
if there are no $\mu$-traces on it, otherwise it is said to be odd.
\end{definition}
Note that the order of applications of the tableau rules is non deterministic
so a formula may appear to have more than one tableau. However, tableau
equivalence, defined next, only looks at the structure of branching,
whether branching nodes are modal or disjunctive, the literals at modal nodes and the parity of
infinite paths, so a formula has a unique tableau, up to tableau equivalence. We define tableau cores
to be the semantic elements of the tableau -- node types, literals at modal nodes, branching structure
and the parity of infinite paths -- which do not depend on the syntax of the generating formula. Finally, we define
trees with back edges which are finite representations of tableau cores.

\begin{definition}
\textit{(Tableau core)} 
A tableau core is $\mathcal{C}=(C,\Omega)$ where $C$ is a potentially infinite but still finitely branching tree of which the nodes are either modal nodes or disjunctive nodes
and modal nodes are decorated with a set of literals. $\Omega$ is a parity assignment with a finite prefix of $\mathbb{N}$ as co-domain. An infinite path in $\mathcal{C}$ is of the parity of
the most significant priority seen infinitely often. 
$\mathcal{C}=(C,\Omega)$ is a tableau core for $\mathcal{T}=(T,L)$ if once the sequences of disjunctions in $\mathcal{T}$ are collapsed into one non-binary disjunction 
there is a bijection $b$ between the branching nodes of $T$ and the nodes of $C$ which respects the following: the successor relation in the sense
that $b(i)$ is a child of $b(j)$ in $C$ if and only if $i$ is a child of $j$ in $T$,
whether nodes are modal or disjunctive, the literals at modal nodes, and the parity of infinite paths. That is to say, if a path in $\mathcal{T}$
maps to a path in $\mathcal{C}$ then the highest priority seen infinitely often on the path in $\mathcal{C}$ is even if and only if the path in $\mathcal{T}$ has no $\mu$-trace.
\end{definition}
\begin{definition}
\textit{(Tableau equivalence)}
Two tableaus $(\mathcal{T}_{0},L_{0})$ and $(\mathcal{T}_{1},L_{0})$
are equivalent if their cores are bisimilar with respect to their
branching structure, whether nodes are disjunctive or modal, the literals at modal nodes and the parity of infinite branches. Two formulas are tableau equivalent if they generate equivalent tableaus.
\end{definition}
\begin{definition}
\textit{(Tree with back edges)} Tableaus are potentially
infinite but regular, so they allow finite representations. A finite representation of a tableau $\mathcal{A}=(A,\Omega)$ is a finite tree with back edges, $A$
which is bisimilar to the core of the tableau. Every node is either a modal node or a disjunctive node and modal nodes are associated with a set of literals.
The tree has a priority assignment $\Omega$ which assigns priorities to nodes such that the highest priority on an infinite path is of the parity of that path.
\end{definition}

To summarise, a tableau $\mathcal{T}$ is a potentially infinite tree labelled with sets of subformulas -- it is specific to the formula which labels its root;
a tableau core, $\mathcal{C}$ is a potentially infinite object which carries the same semantics but is not specific to one formula; finally, a tree with back
edges, called $\mathcal{A}$ because of its resemblance to alternating parity automata, is a finite representation of a tableau core. The next section will present
the one-to-one correspondence between disjunctive formulas and trees with back edges.

\begin{thm}
 \cite{janin1995automata} Tableau equivalent formulas are semantically equivalent.
\end{thm}

Note that tableau equivalence is a stricter notion than semantic equivalence; $\psi \vee \neg \psi$ and $\top$ have different tableau for example.

\subsection{Disjunctive normal form}

Disjunctive form was introduced in \cite{janin1995automata} as a
syntactic restriction on the use of conjunctions. It forces a formula
to follow a simple structure of alternating disjunctions and modalities
where modalities are qualified with a conjunction of propositions. Such
formulas are in many ways well-behaved and easier to manipulate than
arbitrary $L_\mu$ formulas.

\begin{definition}
\textit{(Disjunctive formulas) }The set of disjunctive   form formulas
of (unimodal) $L_{\mu}$ is the smallest set $\mathcal{F}$ satisfying:\end{definition}
\begin{itemize}
\item $\bot,\top$, propositional variables and their negations are in $\mathcal{F}$;
\item If $\psi\in\mathcal{F}$ and $\phi\in\mathcal{F}$ then $\psi\vee\phi\in\mathcal{F}$; 
\item If $\mathcal{A}$ is a set of literals and $\mathcal{B}\subseteq\mathcal{F}$
($\mathcal{B}$ is finite), then $\bigwedge\mathcal{A}\wedge\dmod\mathcal{B}\in\mathcal{F}$
;
\item $\mu X.\psi$ and $\nu X.\psi$ as long as $\psi\in\mathcal{F}$.
\end{itemize}

Every formula is known to be equivalent to an effectively computable
formula in  disjunctive form \cite{janin1995automata}. The transformation into
disjunctive form involves taking the formula's tableau decomposition
and compressing the node labels into a single subformula. The tricky
part is finding a tree with back edges and its priority assignment to represent the tableau finitely, including the parity of infinite paths. The transformation
then turns the tree with back edges into a disjunctive
formula with alternation depth dependent on the priority assignment. Conversely, a disjunctive formula and its minimal
priority assignment induces a tree with back edges representing its tableau.
The minimal priority function required to finitely represent a tableau is therefore equivalent
to the minimal alternation depth of a disjunctive formula generating
the tableau. The following theorem recalls the construction of disjunctive formulas from
trees with back edges labelled with priorities from \cite{janin1995automata}
and shows that the alternation depth of the resulting formula stems from the priority assignment
of the tree with back edges.

\begin{thm}\label{tree-to-formula}
 Let $\mathcal{A}=(A,\Omega)$ be a tree with back edges that is bisimilar to a core of the tableau $\mathcal{T}$ with priority assignment
$\Omega$ with co-domain $\{0...q\}$. Then
there is a disjunctive formula with alternation depth $\{0...q\}$ which generates a tableau equivalent to $\mathcal{T}$.
\end{thm}
\begin{proof}
 First of all, we construct  $\mathcal{A'}=(A',\Omega')$, bisimilar to  $\mathcal{A}$
but with a priority assignment with the following property: on all paths
from root to leaf, the priorities of nodes that are the targets of back edges occur in decreasing order. This is straight-forward
by looking at the infinite tableau core $\mathcal{A}$ unfolds into, remembering which nodes stem from the same node in $\mathcal{A}$ and
their priority assigned by $\Omega$. First consider all branches that see the highest priority $q$ infinitely often and cut them short by creating
back edges at nodes of priority $q$, pointing to the bisimilar ancestor node (also of priority $q$) that is closest to the root. Then repeat this for each priority
in decreasing order, but for each priority $q-1$ treat the ancestor of priority $q$ that back edges point to (if it exists) as the root, so that nodes that have back edges
pointing to them end up in decreasing order of priority. Note that every cycle is now dominated by the priority of the first node from the root seen infinitely often.

 The disjunctive formula is then obtained by assigning a subformula $f(n)$ to every node of $A$ as follows.
If $n$ is a leaf with literals $Q$, then $f(n)=\bigwedge Q$;
if $n$ is a disjunctive node with children $n_{0}$ and $n_{1}$, then $f(n)=f(n_{0})\vee f(n_{1})$;
if $n$ is the source of a back edge of which the target is $m$, then $f(n)=X_{m}$ where $X_{m}$ is a fixpoint variable;
if $n$ is a modal node, then $f(n)=\bigwedge Q\wedge\dmod\mathcal{B}$ where $Q$ is the set
of literals at $n$ and $\mathcal{B}$ is the set of $f(n_{i})$ for $n_{i}$ children of $n$;
other nodes inherit the formula assigned to their unique child.
If $n$ is the target of a back edge, $f(n)$ is obtained as previously detailed but in addition,
it binds the fixpoint variable $X_{n}$ with a $\nu$-binding if $n$ is of even parity and with a $\mu$-binding otherwise.

 If $r$ is the root node of $A'$, then $f(r)$ is
a disjunctive formula that generates a tableau that is equivalent to $\mathcal{T}$. This should be clear from the fact that
the tableau of $f(n)$ consists of the infinite tree generated by $A'$ and the labelling $L(n)=\{f(n)\}$ for all $n$. $\Omega'$ restricted
to the target of back edges is a priority assignment for the disjunctive formula $\Psi=f(n)$ since it respects the
parity of paths and on each branch the priorities occur in decreasing order. This guarantees that if $\Omega'(X)<\Omega'(Y)$
then $X$ is not free in the formula binding $Y$.

Therefore $\Psi$ has a tableau that is equivalent to $\mathcal{T}$ and accepts a priority assignment with co-domain $\{0...q\}$.
\end{proof}

Conversely, a disjunctive formula induces a tree with back edges generating its tableau by taking its tableau
until each branch reaches a fixpoint variable which is the source of a back edge to its binding formula. The priority assignment of the formula is also a priority assignment
for the tree with back-edges. This yields a one-to-one correspondence between trees with back edges and disjunctive formulas.

\section{Tableau equivalence preserves alternation depth for disjunctive $L_{\mu}$}

This section argues that all disjunctive formulas generating the same tableau $\mathcal{T}$
have the same alternation depth. The structures used to identify the alternation depth are
similar to ones found in \cite{rabin-index} to compute the Rabin-Mostowski index of a parity games and
the flowers described in \cite{flowers} to find the Rabin-Mostowski index of non-deterministic automata.
Here I show that tableau equivalence preserves these structures and consequently also the alternation depth of disjunctive formulas. \\

 Definition \ref{witness-def} describes a witness
showing that the priority assignment $\Omega$ of a tree with back edges $\mathcal{A}=(A,\Omega)$ representing
$\mathcal{T}$ requires at least $q$ priorities. This witness is preserved
by bisimulation with respect to node type, literals
and parity of infinite branches. Since all finite representations
of a tableau $\mathcal{T}$ are bisimilar with respect to these criteria,
they all have the same maximal witness, indicating the least number
of priorities $\mathcal{T}$ can be represented with.

Informally, the witness of strictness is a series of cycles of alternating
parity where each cycle is contained within the next.

\begin{definition}\label{witness-def} 
\emph{($q$-witness)} A $q$-witness in a tree with back edges $(A,\Omega)$ representing a tableau $\mathcal{T}$ consists of
$q$ cycles $c_{1}...c_{q}$ such that for each $i\leq q$, the cycle
$c_{i}$ is of the parity of $i$ and for all $0<i<q$, the cycle $c_{i}$ is a subcycle of $c_{i+1}$.
\end{definition}

\begin{lem}\label{priorities}
 If a tree with back edges $(A,\Omega)$ has a $q$-witness, then the co-domain of the priority assignment $\Omega$
has at least $q$ elements.
\end{lem}
\begin{proof}
 Given a $q$-witness $c_{1}...c_q$, for every pair of cycles $c_i$ and $c_{i+1}$, since they are of different parity and $c_i$
is contained in $c_{i+1}$, the dominant priority on $c_{i+1}$ must be strictly larger than the dominant priority on $c_i$. Therefore
there must be at least $q$ priorities in the cycle $c_q$ which contains all the other cycles of the witness.  
\end{proof}

\begin{lem}\label{lem:necessity}
 If a tree with back edges $\mathcal{A}$ representing a tableu $\mathcal{T}$ 
does not have a $q$-witness, then there is an tree with back edges $\mathcal{A}'$ which
also represents $\mathcal{T}$ but has a priority assignment with fewer priorities.\end{lem}
\begin{proof}
Assume a tree with back edges $\mathcal{A}=(A,\Omega)$ representing $\mathcal{T}$
with a priority assignment with co-domain $\{0...q\}$ does not have a $q$ witness. Let $S_{q}$
be the set of nodes of priority $q$. Let $S_{i-1}$ for $1<i\leq q$
be the set of nodes of priority $i-1$ which appear as the second
highest priority in a cycle where all the nodes of highest priority
are in $S_{i}$, and as the nodes of highest priority in some cycle. Note that if $S_1$ was non-empty, then there would
be a $q$-witness, so $S_1$ and consequently $S_0$ must be empty. Then define a new priority
function as follows: the new priority function $\Omega'$ is as $\Omega$, except
for nodes in any $S_{i}$ -- these receive the priority $i-2$ instead
of the priority $i$. Since $S_{1}$ and $S_{0}$ are empty, this
is possible whilst keeping all priorities positive. 
 $\Omega'$ with co-domain $\{0...q-1\}$ preserves
the parity of infinite branches since there are no cycles in which
the priority of all dominant nodes is decreased more than the priority
of all sub-dominant nodes and each node retains the same parity. Therefore,
if a finite representation of $\mathcal{T}$ does not have a $q$-witness, then there is
a finite representation $\mathcal{A}'=(A,\Omega')$ with a smaller priority assignment. \end{proof}
\begin{lem}\label{sufficiency}
 All tableau equivalent trees with back edges have the same $q$-witnesses: for all $q$, either all or
none of the trees with back edges representing a same tableau $\mathcal{T}$ have a $q$-witness.
 \end{lem}
\begin{proof}
First we recall that if $\mathcal{A}$ is the finite representation
of $\mathcal{T}$ induced by a disjunctive formula $\Psi$ then the
tableau of $\mathcal{T}$  is an infinite tree bisimilar to $\mathcal{A}$
with respect to node type, literals and parity of infinite
branches. Hence any finite representation of $\mathcal{T}$ is bisimilar
to $\mathcal{A}$. It then suffices to show that
$q$-witnesses are preserved under bisimulation. This is straight-forward:
 let $\mathcal{A}'$ be bisimilar to a finite tree with back edges $\mathcal{A}$ with respect to node type,
literals at modal nodes and the parity of infinite paths. Then
infinite paths in $\mathcal{A}$ are bisimilar to infinite paths in
$\mathcal{A}'$. Since both $\mathcal{A}$ and $\mathcal{A}'$ are
finite, an infinite path stemming from a cycle in $\mathcal{A}$ is
bisimilar to a cycle in $\mathcal{A}'$. 
A $q$-witness contains at least one node which lies on all the cycles of the witness.
If $\mathcal{A}$ has  $q$ cycles, call the node on all of its cycles $n$ and
consider (one of) the deepest node(s) $n'$ in $\mathcal{A}'$ bisimilar to $n$. That is to say, choose
$n'$ such that if another node bisimilar to $n'$ is reachable from $n'$, it must be an ancestor of $n'$.
Since $n'$ is bisimilar to $n$, there must be a cycle $c_i'$ bisimilar to each $c_i$ reachable from $n'$.
Since $n'$ is maximally deep, it is contained in each of these cycles $c_i'$. Then, a
$q$-witness can be reconstructed in $\mathcal{A}'$ by taking the
cycle $c_1'$, and then for each $i>0$ the cycle consisting of all $c_j', j\leq i$ . Since all $c_i'$ cycles have
$n'$ in common, there is a cycle combining $c_j', j\leq i$ for any $i$.
 Since bisimulation
respects the parity of cycles, this yields a $q$-witness in $\mathcal{A}'$.\end{proof}
% \begin{definition}
% Let the disjunctive alternation depth of a tableau $\mathcal{T}$
% be the least range $[0,...,q]$ of any disjunctive alternating parity
% automaton with tableau $\mathcal{T}$.\end{definition}
\begin{thm}\label{thm:same}
All disjunctive formulas with tableau $\mathcal{T}$ have the same alternation depth.\end{thm}
\begin{proof}
All trees with back edges representing the same tableau $\mathcal{T}$ have the same maximal witness, from the previous lemma, so from Lemma \ref{lem:necessity}
they accept a minimal priority function with domain $\{0...q\}$.
Since a disjunctive formula induces a tree with back edges with a minimal priority function corresponding to the formula's alternation depth,
any two disjunctive formulas that are tableau equivalent must have the same alternation depth.
\end{proof}

This concludes the proof that tableau equivalence preserves alternation depth on disjunctive formulas.
The restriction to disjunctive formulas is crucial: as the next section shows, in the general case tableau
equivalent formulas may have vastly different alternation depths.

\section{Disjunctive form does not preserve alternation depth}

Every formula has a tableau which allows it to be turned into a semantically
equivalent disjunctive formula. This section studies the relationship
between a formula's alternation depth and the alternation
depth of its tableau equivalent disjunctive form. As the previous
section shows, any two disjunctive formulas with the same tableau
have the same alternation depth; therefore comparing a non-disjunctive
formula to any tableau equivalent disjunctive formula will do.

The first subsection demonstrates that not only does disjunctive
form not preserve alternation depth, but also that there is no hope
for bounding the alternation depth of disjunctive formulas with respect
to their semantic alternation depth: for any $n$ there are one alternation
formulas which are tableau equivalent to $n$ alternation disjunctive formulas.
In other words, the alternation depth of a $L_\mu$ formula, when transformed
into disjunctive form, can be arbitrarily large. Conversely, as shown in
the second subsection, formulas of arbitrarily large alternation depth
can be tableau equivalent to a disjunctive formula without alternations.
Hence the alternation depth of tableau equivalent formulas are only related
within the disjunctive fragment.

\subsection{Disjunctive formulas with large alternation depth}

While the main theorem is proved by Example \ref{ex-3}, the Examples \ref{ex:simple} and
\ref{exa:ex-2} leading up to it should give the interested reader some intuition about the
mechanics which lead the tableau of a formula to have higher alternation
depth than one might expect.
\begin{example}\label{ex:simple}
The first example is a rather simple one: a disjunctive formula with
one alternation that can be expressed in non-disjunctive form without
any alternations. The disjunctive formula $\nu X.\mu Y.(A\wedge\dmod\{X\})\vee(\bar{A}\wedge\dmod\{Y\})$
signifies that all paths are infinite and $A$ occurs infinitely often
on all paths. Compare it to the formula $\nu X.\dmod\{X\}\wedge\mu Y.(\bar{A}\wedge\dmod\{Y\})\vee A$ which is alternation free.

The tableaus of both these formulas are shown side by side in Figure \ref{fig:simple}.
Both branches regenerate into either exactly the ancestral node marked * or a node that
reaches a node identical to the one marked * in a single non branching step.\\
The cores of the two tableaus, that is to say their branching nodes, are clearly isomorphic with respect to
the node type and branching structure. Furthermore, for both formulas, there is $\mu$-trace on any path that
only goes through the left hand branch infinitely often. There is no $\mu$ trace on \emph{any} path that goes
through the right hand path infinitely often, for either formula. As a result, both tableaus agree on the parity of infinite branches.
The two formulas are tableau equivalent and therefore also semantically equivalent.
\end{example}

\begin{remark}
  Observe that there is nothing obviously inefficient about how the disjunctive formula handles alternations.
Indeed, simply inverting the order of the fixpoints yields a formulas
which can not be expressed without an alternation: $\mu X.\nu Y.A\wedge\dmod\{X\}\vee\bar{A}\wedge\dmod\{Y\}$. 
\end{remark}

\begin{figure}\label{fig:simple}

\begin{prooftree}

\AxiomC{\textbf{*}}
\UnaryInfC{$Y,X$\hspace{4mm}
$\mathbf{W}$}

\UnaryInfC{$\dmod \{Y\}, \bar{A}, \dmod \{X\}$\hspace{4mm}
$\mathbf{ \bar{A}, \dmod \{W\}}$}
\UnaryInfC{$\dmod \{Y\}, \bar{A}\wedge \dmod \{X\}$\hspace{4mm}
$\mathbf{ \bar{A}\wedge \dmod \{W\}}$}

\AxiomC{\textbf{*}}
\UnaryInfC{$Y$\hspace{4mm}
$\mathbf{Z}$} 

\UnaryInfC{$\dmod \{Y\}, A$\hspace{4mm}
$\mathbf{A, \dmod\{Z\}}$} 

\UnaryInfC{$\dmod \{Y\}, A$ \hspace{4mm}
$\mathbf{A\wedge \dmod\{Z\}}$}

\BinaryInfC{\textbf{*}\hspace{3mm}$\dmod \{Y\} , (\bar{A} \wedge \dmod \{X\}) \vee A$ \hspace{4mm}
$\mathbf{(\bar{A}\wedge \dmod \{W\}) \vee (A\wedge \dmod\{Z\})}$}
\UnaryInfC{$\nu Y.\dmod \{Y\} \wedge \mu X. (\bar{A} \wedge \dmod \{X\}) \vee A$ \hspace{4mm}
$\mathbf{\nu Z.\mu W. (\bar{A}\wedge \dmod \{W\}) \vee (A\wedge \dmod\{Z\})}$}

\end{prooftree}

\caption{Tableaus for $\nu Y.\dmod \{Y\} \wedge \mu X. (\bar{A} \wedge \dmod \{X\}) \vee A$ and $\mathbf{\nu Z.\mu W.  (\bar{A}\wedge \dmod \{W\}) \vee (A\wedge \dmod\{Z\})}$}

\end{figure}

While the above example proves that disjunctive form does not preserve
alternation, it must be noted that the alternating parity automata
corresponding to these formulas require in both cases two priorities,
although only one requires an alternation. The next example shows
formulas in which the number of priorities is not preserved either.

\begin{example}\label{exa:ex-2}
This example and the following ones will be built
on one-alternation formulas consisting of single $\mu/\nu$ alternations
embedded in one another without interfering with each other, i.e. all
free variables within the inner formula $\phi_{1}$ are bound by the
inner fixpoint bindings. This means that the formula accepts a priority
assignment with co-domain $\{0,1\}$. Without further ado,
consider the formula in question:
$$\alpha=\mu X_{0}.\nu Y_{0}.(A\wedge\dmod\{X_{0}\})\vee(B\wedge\dmod\{Y_{0})\wedge\mu X_{1}.\nu Y_{1}.(C\wedge\dmod\{X_{1}\})\vee(D\wedge\dmod\{Y_{1}\})\vee E$$
\end{example}
The following Lemma shows it to be equivalent to a formula which requires a priority assignment with co-domain $\{0...3\}$.
\begin{lem}\label{lem:equivalence}
The formula
$\alpha$ is tableau equivalent to a disjunctive formula which requires a parity assignment with co-domain $\{0...3\}$:
\begin{equation}
\begin{split}
\beta = \mu X_{0}.\nu Y_{0}.\mu X_{1}.\nu Y_{1}.(A\wedge C\wedge\dmod\{X_{0}\})\vee(A\wedge D\wedge\dmod\{X_{0}\})\vee(A\wedge E\wedge\dmod\{X_{0}\}) \\
\vee(B\wedge E\wedge\dmod\{Y_{0}\})\vee(B\wedge C\wedge\dmod\{X_{1}\})\vee(B\wedge D\wedge\dmod\{Y_{1}\})
\end{split}
\end{equation}
\end{lem}

\begin{figure}

\footnotesize{
\begin{scprooftree}{0.9}

\AxiomC{\textbf{*}}
\UnaryInfC{{$Y_0,Y_1$}}
\UnaryInfC{\tiny{$(B, \dmod \{Y_0\}, D,\dmod\{Y_1\})$}}
\UnaryInfC{\tiny{$(B\wedge \dmod \{Y_0\}), (D\wedge \dmod \{Y_1\})$}} 

\AxiomC{\textbf{*}}
\UnaryInfC{{$Y_0,X_1$}}
\UnaryInfC{\tiny{$(B, \dmod \{Y_0\}, C,\dmod\{X_1\})$}} 
\UnaryInfC{\tiny{$(B\wedge \dmod \{Y_0\}), (C\wedge \dmod \{X_1\})$}} 

\AxiomC{\textbf{*}}
\UnaryInfC{{$Y_0$}}
\UnaryInfC{\tiny{$(B, \dmod \{Y_0\}, E)$}} 
\UnaryInfC{\tiny{$(B\wedge \dmod \{Y_0\}, E)$}} 

\TrinaryInfC{$(B\wedge \dmod \{Y_0\}) ,  (C\wedge \dmod \{X_1\})\vee (D\wedge \dmod \{Y_1\})\vee E$}

\AxiomC{\textbf{*}}
\UnaryInfC{{$X_0,Y_1$}}
\UnaryInfC{\tiny{$(A, \dmod \{X_0\}, D,\dmod\{Y_1\})$}}
\UnaryInfC{\tiny{$(A\wedge \dmod \{X_0\}), (D\wedge \dmod \{Y_1\})$}}

\AxiomC{\textbf{*}}
\UnaryInfC{{$X_0,X_1$ }}
\UnaryInfC{\tiny{$(A, \dmod \{X_0\}, C,\dmod\{X_1\})$} }
\UnaryInfC{\tiny{$(A\wedge \dmod \{X_0\}), (C\wedge \dmod \{X_1\})$} }

\AxiomC{\textbf{*}}
\UnaryInfC{{$X_0$}}
\UnaryInfC{\tiny{$(A, \dmod \{X_0\}, E)$} }
\UnaryInfC{\tiny{$(A\wedge \dmod \{X_0\}, E)$} }

\TrinaryInfC{$(A\wedge \dmod \{X_0\}) ,  (C\wedge \dmod \{X_1\})\vee (D\wedge \dmod \{Y_1\})\vee E$}

\BinaryInfC{\textbf{*}$(A\wedge \dmod \{X_0\})\vee (B\wedge \dmod \{Y_0\}) ,  (C\wedge \dmod \{X_1\})\vee (D\wedge \dmod \{Y_1\})\vee E$}

\UnaryInfC{$\mu X_0. \nu Y_0.(A\wedge \dmod \{X_0\})\vee (B\wedge \dmod \{Y_0\}) \wedge \mu X_1 \nu Y_1 (C\wedge \dmod \{X_1\})\vee (D\wedge \dmod \{Y_1\})\vee E$}
\end{scprooftree}
}
\caption{Tableau for $\alpha$}\label{fig:alpha}

\footnotesize{
\begin{scprooftree}{0.9}

\AxiomC{\textbf{*}}
\UnaryInfC{$Y_1$}
\UnaryInfC{\scriptsize{$ (B, D, \dmod \{Y_1\})$}}
\UnaryInfC{\scriptsize{$ (B\wedge D\wedge \dmod \{Y_1\})$}}

\AxiomC{\textbf{*}}
\UnaryInfC{$X_1$}
\UnaryInfC{\scriptsize{$ (B, C, \dmod \{X_1\})$}}
\UnaryInfC{\scriptsize{$ (B\wedge C\wedge \dmod \{X_1\})$}}

\AxiomC{\textbf{*}}
\UnaryInfC{$Y_0$}
\UnaryInfC{\scriptsize{$ (B, E, \dmod \{Y_0\})$}}
\UnaryInfC{\scriptsize{$ (B\wedge E\wedge \dmod \{Y_0\})$}}

\TrinaryInfC{$ (B\wedge E\wedge \dmod \{Y_0\})\vee (B\wedge C\wedge \dmod \{X_1\})\vee (B\wedge D\wedge \dmod \{Y_1\})$}

\AxiomC{\textbf{*}}
\UnaryInfC{$X_0$}
\UnaryInfC{\scriptsize{$ (A, D, \dmod \{X_0\})$}}
\UnaryInfC{\scriptsize{$ (A\wedge D\wedge \dmod \{X_0\})$}}

\AxiomC{\textbf{*}}
\UnaryInfC{$X_0$}
\UnaryInfC{\scriptsize{$ (A, C, \dmod \{X_0\})$} }
\UnaryInfC{\scriptsize{$ (A\wedge C\wedge \dmod \{X_0\})$}}

\AxiomC{\textbf{*}}
\UnaryInfC{$X_0$}
\UnaryInfC{\scriptsize{$ (A, E, \dmod \{X_0\})$} }
\UnaryInfC{\scriptsize{$ (A\wedge E\wedge \dmod \{X_0\})$}}

\TrinaryInfC{$ (A\wedge E\wedge \dmod \{X_0\}) \vee (A\wedge D\wedge \dmod \{X_0\})\vee (A\wedge C\wedge \dmod \{X_0\})$}

\BinaryInfC{\textbf{*}
$ (A\wedge E\wedge \dmod \{X_0\}) \vee (A\wedge D\wedge \dmod \{X_0\})\vee (A\wedge C\wedge \dmod \{X_0\})\vee 
(B\wedge E\wedge \dmod \{Y_0\})\vee (B\wedge C\wedge \dmod \{X_1\})\vee (B\wedge D\wedge \dmod \{Y_1\})$}

\UnaryInfC{\footnotesize{$
\mu X_0. \nu Y_0. \mu X_1. \nu Y_1. (A\wedge E\wedge \dmod \{X_0\}) \vee (A\wedge D\wedge \dmod \{X_0\})\vee (A\wedge C\wedge \dmod \{X_0\})\vee
(B\wedge E\wedge \dmod \{Y_0\})\vee (B\wedge C\wedge \dmod \{X_1\})\vee (B\wedge D\wedge \dmod \{Y_1\})$}}
\end{scprooftree}
}
\caption{Tableau for $\beta$}\label{fig:beta}

\end{figure}

\begin{proof}
The tableaus for both formulas are written out in Figures \ref{fig:alpha} and \ref{fig:beta}.
The two tableaus are isomorphic with respect to branching structure,
node type and the literals at modal nodes. To prove
their equivalence, it is therefore sufficient to argue that this isomorphism
also preserves the parity of infinite branches, that is to say that there is a $\mu$-trace in an infinite path of
one if and only if there is a
$\mu$-trace in the corresponding infinite path of the other. \\

To do so, we look, case by case, at the combinations of branches
that a path can see infinitely often and check which have a $\mu$ trace in each tableau.
First argue that the three right-most branches in both tableaus are such that
any path that sees them infinitely often has a $\mu$-trace. This
is witnessed in both cases by the least fixpoint variable $X_{0}$
which will dominate any trace it appears on and appears on a trace
on all paths going through one of these branches infinitely often.
So, in both tableaus, any path going through one of the right-most
branches infinitely often is of odd parity. Now consider the branch
that ends in $Y_{0}$ before regenerating to the node marked * in both tableaus.
All traces on paths that go infinitely often through this branch will
see $Y_{0}$ regenerate infinitely often. Therefore in both tableaus,
a path going through this branch infinitely has a $\mu$ trace if
and only if it also goes through one of the three rightmost branches
infinitely often. Now consider the fifth branch from the right, the
branch that regenerates $Y_{0},X_{1}$ in one case and just $X_{1}$
in the other. In both tableaus, a path that goes through this branch
infinitely often will have a $\mu$ trace unless it goes through the
$Y_{0}$ branch infinitely often and doesn't go through one of the
three right-most branches infinitely often. Finally, in both tableaus,
a branch that only sees the left-most branch infinitely often is of
even parity since such a path does not admit any $\mu$-traces. However,
if a path sees this branch and some other branches infinitely often,
its parity is determined by one of the previously analysed cases.
Since we have analysed all the infinite paths on these tableaus and concluded that
in each case the parity of a path is the same in both tableaus, this concludes
the proof that the two tableaus are equivalent.
\end{proof}

The above example yields a disjunctive formula of alternation depth $\{0...3\}$
which semantically only requires alternation depth $\{0,1\}$. This proves that disjunctive
form does not preserve the number of priorities the model checking game of a formula requires.

The next step is to generalise the construction of Example \ref{exa:ex-2}
to arbitrarily many alternations to prove that there is no bound on the number of
alternations of a disjunctive formula tableau equivalent to a non-disjunctive formula of $n$
alternations. To do so, we will first define the one-alternation
formulas $\alpha_{n}$ inductively, based on the formula of Example
\ref{exa:ex-2}. We then argue that the tableau of $\alpha_{n}$ admits
a $(2n+1)$-witness, proving that $\alpha_{n}$ is not tableau equivalent
to any disjunctive formula of less than $2n+1$ alternations. Due
to the argument pertaining to traces in increasingly large tableaus,
its details are, inevitably, quite involved. However, the mechanics
of the tableaus of $\alpha_{n}$ are not difficult; writing down the
tableau of $\alpha_{2}$ and working out its disjunctive form should suffice
to gain an intuition of the proof to follow.
\begin{example}\label{ex-3}
In order to define $\alpha_{n}$ for any $n$ define:\\
\begin{equation}
\begin{split}
a_{1}=\mu X_{1}.\nu Y_{1}.((A_{1}\wedge\dmod\{X_{1}\})\vee(B_{1}\wedge\dmod\{Y_{1}\})\vee E_{1})\wedge \\
\mu X_{0}.\nu Y_{0}.(A_{0}\wedge\dmod\{X_{0}\})\vee(B_{0}\wedge\dmod\{Y_{0}\})\vee E_{0} \\
a_{i+1}=\mu X_{i+1}.\nu Y_{i+1}.((A_{i+1}\wedge\dmod\{X_{i+1}\}) \vee(B_{i+1}\wedge\dmod\{Y_{i+1}\})\vee E_{i+1})\wedge a_{i}
\end{split}
\end{equation}

Then, define: $$\alpha_{n}=\mu X_{n}.\nu Y_{n}.((A_{n}\wedge\dmod\{X_{n}\})\vee(B_{n}\wedge\dmod\{Y_{n}\}))\wedge a_{n-1}$$
 In other words, the formula consists of nested clauses $\mu X_{i}.\nu Y_{i}.((A_{i}\wedge\dmod\{X_{i}\})\vee(B_{i}\wedge\dmod\{Y_{i}\})\vee E_{i})$
connected by conjunctions where the outmost clause does not have a
$\vee E$.
\end{example}

As the formula grows, its tableau becomes unwieldy, but its structure
remains constant: it is just as the tableau of $\alpha$ with more branches.
 Figure \ref{fig:alpha} can be used as reference.

The tableau of any $\alpha_{n}$ follows this structure:
\begin{itemize}
\item The first choice node $\{(A_{n}\wedge\dmod\{X_{n}\}\vee B_{n}\wedge\dmod\{Y_{n}\}),...,(A_{0}\wedge\dmod\{X_{0}\}\vee B_{0}\wedge\dmod\{Y_{0}\}\vee E_{0})\}$
branches into $2\times3^{n}$ modal nodes -- ignoring the modalities attached
to each literals for a moment, this is the decomposition of
$(A_{n}\vee B_{n})\wedge(A_{n-1}\vee B_{n-1}\vee E_{n-1})...\wedge(A_{0}\vee B_{0}\vee E_{0})$
into one large disjunction.
\item Each choice leads to a modal node with some choice of propositional
variables consisting of one of $A_{n}$ and $B_{n}$ and then for
every $i<n$ one of $A_{i},B_{i}$ or $E_{i}$.
\item These modal nodes have a single successor each, consisting of a set
of fixpoint variables. In every case, one of these is $Y_{n}$ or
$X_{n}$ and there is only ever at most one fixpoint variable out
of $\{X_{i},Y_{i}\}$ for each $i$. These nodes will be referred to
as regeneration nodes. When a regeneration node does not contain $X_i$ nor $Y_i$
for some $i$, this corresponds to $E_i$ having been chosen rather than $A_i$ or $B_i$.
\item Nodes consisting of a set of fixpoint variables all regenerate, give
or take a couple of non-branching steps, into the same choice node, identical to the ancestral choice node labelled:
$$\{(A_{n}\wedge\dmod\{X_{n}\}\vee B_{n}\wedge\dmod\{Y_{n}\}),...,(A_{0}\wedge\dmod\{X_{0}\}\vee B_{0}\wedge\dmod\{Y_{0}\}\vee E_{0})\}$$

\item An infinite trace in this tableau sees infinitely often only fixpoint
variables $Y_{i}$ and/or $X_{i}$ for some $i$. As a consequence if
a path goes infinitely often through a regeneration node which does
not contain $X_{i}$ or $Y_{i}$, then there is no trace that sees
$X_{i}$ infinitely often on that path.
\end{itemize}

\begin{lem}
 The formula $\alpha_n$ is tableau equivalent only to disjunctive formulas which require a priority assignment with $2n+1$ priorities. 
\end{lem}

\begin{proof}
Using the above observations, we will show that the tableau for this
formula requires at least $2n +1$ alternating fixpoints. We describe a priority
assignment to a subset of the nodes of the tableau of $\alpha_{n}$
such that on the paths within this subset, a path is even if and only
if the most significant priority seen infinitely often is even. 
We then argue that this subset constitutes a $2n+1$-witness.

Consider the paths of the tableau which only contain the following regeneration nodes:
\begin{itemize}
 \item For all $i$, the nodes regenerating exactly $Y_{n}Y_{n-1}...Y_{i}$, and
 \item For all $i$ the nodes regenerating exactly $Y_{n}...Y_{i+1}X_{i}Y_{i-1}...Y_{0}$.
\end{itemize}

For each $i$, assign priority $2i$ to the node regenerating $Y_{n}...Y_{i}$ and $2i+1$ to the node
regenerating $Y_{n}...Y_{i+1},X_{i},Y_{i-1},...Y_{0}$. We now prove that this priority assignment
is such that a path within this subtableau is even if and only if the highest priority seen infinitely often
is even. 

First consider the nodes $Y_{n}...Y_{i}$, which have been assigned even priority.
 A path that sees such a node infinitely often can only have a $\mu$-trace if it sees
a node regerating some $X_{j}$, $j>i$ infinitely often. Such a node would have an odd
priority greater than $Y_{n}...Y_{i}$. Therefore, if the most significant priority
seen infinitely often is even, the path has no $\mu$ trace.
 Conversely, if a path sees $Y_{n}...X_{i}...Y_{0}$ infinitely often and no
$Y_{n}...Y_{j}$ where $j>i$ infinitely often, then there is a trace which only
regenerated $X_{i}$ and $Y_{i}$ infinitely often. This is a $\mu$ trace since $X_{i}$
is more significant than $Y_{i}$. This priority assignment therefore describes the parity
of infinite paths on this subset of paths of $\mathcal{T}$.

Any assignment of priorities onto $\mathcal{T}$ should, on this
subset of paths, agree in parity with the above priority assignment.
However, in any tree with back edges generating this tableau, this
subset of paths constitutes a $2n+1$ witness: $c_{0}$ is a cycle
that only sees $Y_{n}...Y_{0}$, $c_{1}$ contains $c_{0}$ and also sees $Y_{n}...X_{1}Y_{0}$ infinitely often
and for all $i>1$, the cycle $c_{2i}$ is one containing $c_{2i-1}$ and $Y_{n}...Y_{i}$ while $c_{2i+1}$ is one
containing $c_{2i}$ and $Y_{n}...X_{i}...Y_{0}$. Each cycle $c_j$ is dominated by the priority $j$, making $c_0,...,c_{2i+1}$
a $2i+1$-witness. 
 Thus, using Theorem \ref{thm:same} any disjunctive formula with tableau $\mathcal{T}$ must require at least $2n+1$ priorities.
\end{proof}

This concludes the proof that for arbitrary $n$, there are one-alternation $L_\mu$ formulas which
are tableau equivalent to disjunctive formulas with $n$ alternations. 
\subsection{Disjunctive formulas with small alternation depth}

The previous section showed that transforming a formula into disjunctive form can increase its alternation depth. The converse is
much easier to show: there are very simple formulas for which the transformation into disjunctive form eliminates all alternations.

\begin{lem}\label{lem:finite}
For any formula $\psi$, the formula $(\mu X.\dmod\{X\}\vee\dmod\bot)\wedge\psi$
is tableau equivalent to a disjunctive formula without $\nu$-operators.\end{lem}
\begin{proof}
The semantics of $(\mu X.\dmod\{X\}\vee\dmod\bot)\wedge\psi$ are
that a structure must not have infinite paths and $\psi$ must hold.
Consider $\mathcal{T}$ , the tableau for $(\mu X.\dmod\{X\}\vee\dmod\bot)\wedge\psi$.
It is easy to see that every modal node will either contain $\dmod\{X\}$
or $\dmod\bot$. The latter case terminates that branch of the tableau,
while the former will populate every successor node with $X$ which
will then regenerate into $(\dmod\{X\}\vee\dmod\bot)$. As
a result, all infinite paths have a $\mu$ trace; there are no even
infinite paths. Any disjunctive formula generating $\mathcal{T}$
will therefore only require the $\mu$ operator.
\end{proof}
Taking $\psi$ to be a formula of arbitrarily high alternation depth,
$(\mu X.\dmod\{X\}\vee\dmod\bot)\wedge\psi$ shows that the transformation
into disjunctive form can reduce the alternation depth an arbitrarily
large amount. Together with the previous section,
this concludes the argument that there are no bounds on the difference in alternation depth of tableau equivalent formulas.

\section{Discussion}

To summarise, we have studied how tableau decomposition and the transformation into disjunctive form affects
the alternation depth of a formula. The first observation is that within the confines of the disjunctive fragment
of $L_\mu$, alternation depth is very well-behaved with respect to tableau equivalence: any two tableau equivalent
disjunctive formulas have the same alternation depth. However, the story is quite different for $L_\mu$ without the restriction
to disjunctive form: the alternation depth of a $L_\mu$ formula can not be used to predict any bounds on the alternation
depth of tableau equivalent disjunctive formulas and vice versa.\\

Part of the significance of this result are the implications for our understanding of the alternation hierarchy.

The formulas in Section 4 illutrate some of the different types of accidental complexity which any
procedure for deciding the alternation hierarchy would need to somehow overcome. The formula $(\mu X.\dmod \{X\}\vee\bot)  \wedge \psi$, from
Lemma \ref{lem:finite}  which is semantically a $\nu$-free formula for any $\psi$ is an example of a type of accidental complexity which the
tableau decomposition eliminates. However, the formula in Example \ref{ex:simple} illustrate a more subtle form of
accidental complexity that is immune to disjunctive form:  $\nu X.\mu Y.(A\wedge\dmod\{X\})\vee(\bar{A}\wedge\dmod\{Y\})$
is semantically alternation free while the syntactically almost identical formula $\mu X.\nu Y.(A\wedge\dmod\{X\})\vee(\bar{A}\wedge\dmod\{Y\})$ is not. These
formulas pinpoint a very specific challenge facing algorithms that try to reduce the alternation depth of formulas; as such,
they are valuable case studies for those seeking to understand the $L_\mu$ alternation hierarchy.

 Finally, we showed that the following is decidable: for any $L_\mu$ formula, the least alternation depth of a tableau equivalent disjunctive
formula is decidable. This raises the question of whether the same is true if we lift the restriction to disjunctive form,
but keep the restriction to tableau equivalence: for a $L_\mu$ formula, is the least alternation depth of any tableau equivalent
formula decidable? Tableau equivalence is a stricter equivalence to semantic equivalence, so this problem is likely to be easier than
deciding the alternation hierarchy with respect to semantic equivalence but it would still be a considerable step towards understanding accidental
complexity in $L_\mu$.

\paragraph{Acknowledgements}
I thank the anonymous reviewers for their thoughtful comments, which have helped improve the presentation of
this paper and relate this work to similar results for other automata.

\section{Bibliography}

\bibliographystyle{eptcs}
\bibliography{fics2015}
\end{document}